\documentclass[pdflatex,sn-basic]{sn-jnl}

\jyear{2021}%

\theoremstyle{thmstyleone}%
\newtheorem{theorem}{Theorem}
\newtheorem{lemma}{Lemma}

\theoremstyle{thmstyletwo}%
\newtheorem{remark}{Remark}%

\theoremstyle{thmstylethree}%
\newtheorem{definition}{Definition}%

\begin{document}

\title[Dollar Cost Averaging Returns Estimation]{Dollar Cost Averaging Returns Estimation}


\author{Hayden Brown\footnote{Department of Mathematics and Statistics, University of Nevada, Reno\\Email address: haydenb@nevada.unr.edu\\ORCID: 0000-0002-2975-2711}}




\abstract{Given a geometric Brownian motion wealth process, a log-Normal lower bound is constructed for the returns of a regular investing schedule. The distribution parameters of this bound are computed recursively. For dollar cost averaging (equal amounts in equal time intervals), parameters are computed in closed form. A lump sum (single amount at time 0) investing schedule is described which achieves a terminal wealth distribution that matches the wealth distribution indicated by the lower bound. Results are applied to annual returns of the S\&P Composite Index from the last 150 years. Among data analysis results, the probability of negative returns is less than 2.5\% when annual dollar cost averaging lasts over 40 years.}





\pacs[Acknowledgments]{I am grateful to Andrey Sarantsev for helpful remarks about early drafts of this manuscript.}

\maketitle

\section{Introduction}\label{sec1}
Consider an investing schedule where investment amounts and times are predetermined. When the price process is geometric Brownian motion, the return distribution for such an investing schedule lacks simple expression. In light of this shortcoming, a log-Normal lower bound on the returns is constructed. Then the lower bound is applied to dollar cost averaging (DCA) and lump sum (LS) investing using historic price data from the S\&P Composite Index. 

Figure \ref{fig:genpic} illustrates the general situation where investment amounts and times are predetermined. DCA and LS are special cases of this general situation. DCA invests a constant amount at equidistant time steps. LS invests a lump sum at the initial time, without making any additional investments. In what follows, note that return indicates terminal wealth divided by the total invested.

\begin{figure}[h] 
\setlength{\unitlength}{0.4cm}
\begin{picture}(20,18)
\thicklines
\put(4,17){General}
\put(1,7){\begin{turn}{90}\text{time}\end{turn}}
\put(4,1){\vector(0,1){15}}
\put(4,1){\line(-1,0){.4}}
\put(2.8,.7){0}
\put(5,1){\vector(-1,0){1}}
\put(5,.8){{\footnotesize buy $\frac{c_0}{X(0)}$ shares}}
\put(4,3){\line(-1,0){.4}}
\put(2.8,2.7){$t_1$}
\put(5,3){\vector(-1,0){1}}
\put(5,2.8){{\footnotesize buy $\frac{c_1}{X(t_1)}$ shares}}
\put(4,8){\line(-1,0){.4}}
\put(2.8,7.7){$t_2$}
\put(5,8){\vector(-1,0){1}}
\put(5,7.8){{\footnotesize buy $\frac{c_2}{X(t_2)}$ shares}}
\put(4,11){\line(-1,0){.4}}
\put(2.8,10.7){$t_3$}
\put(5,11){\vector(-1,0){1}}
\put(5,10.8){{\footnotesize buy $\frac{c_3}{X(t_3)}$ shares}}
\put(4,13){\line(-1,0){.4}}
\put(2.8,12.7){$t_4$}
\put(5,13){\vector(-1,0){1}}
\put(5,12.8){{\footnotesize buy $\frac{c_4}{X(t_4)}$ shares}}

\put(14,17){DCA}
\put(14,1){\vector(0,1){15}}
\put(14,1){\line(-1,0){.4}}
\put(12.8,.7){0}
\put(15,1){\vector(-1,0){1}}
\put(15,.8){{\footnotesize buy $\frac{c_0}{X(0)}$ shares}}
\put(14,4){\line(-1,0){.4}}
\put(12.8,3.7){$1$}
\put(15,4){\vector(-1,0){1}}
\put(15,3.8){{\footnotesize buy $\frac{c_0}{X(1)}$ shares}}
\put(14,7){\line(-1,0){.4}}
\put(12.8,6.7){$2$}
\put(15,7){\vector(-1,0){1}}
\put(15,6.8){{\footnotesize buy $\frac{c_0}{X(2)}$ shares}}
\put(14,10){\line(-1,0){.4}}
\put(12.8,9.7){$3$}
\put(15,10){\vector(-1,0){1}}
\put(15,9.8){{\footnotesize buy $\frac{c_0}{X(3)}$ shares}}
\put(14,13){\line(-1,0){.4}}
\put(12.8,12.7){$4$}
\put(15,13){\vector(-1,0){1}}
\put(15,12.8){{\footnotesize buy $\frac{c_0}{X(4)}$ shares}}

\put(23.5,17){LS}
\put(23.5,1){\vector(0,1){15}}
\put(23.5,1){\line(-1,0){.4}}
\put(22.3,.7){0}
\put(24.5,1){\vector(-1,0){1}}
\put(24.5,.8){{\footnotesize buy $\frac{c_0}{X(0)}$ shares}}
\end{picture}
\caption{Illustrates investment in a particular stock using predetermined amounts and times. $X:[0,\infty)\to(0,\infty)$ is the share price, and the $c_k$ are positive constants. The ability to buy fractional shares is assumed.}
\label{fig:genpic}
\end{figure}
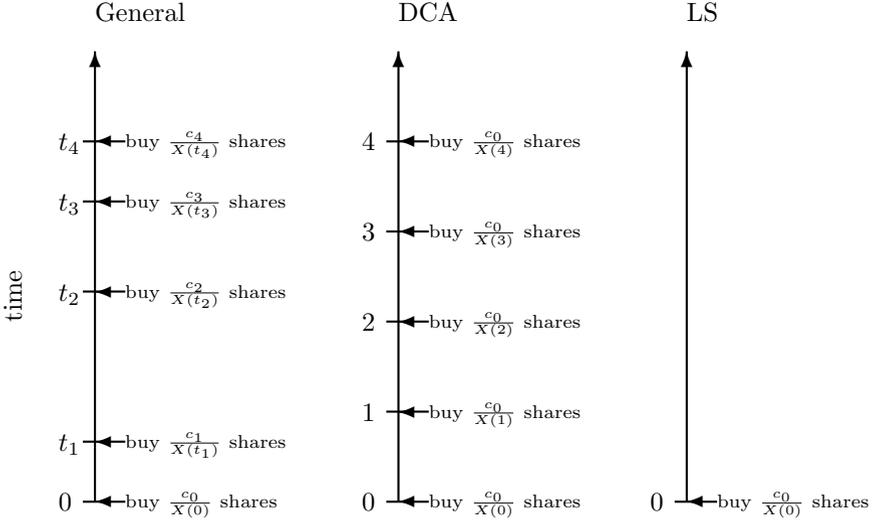

\subsection{Literature Review}
Much of the research on investing schedules with predetermined amounts and times focuses on comparing DCA with LS, provided the total amount invested and total length time invested are equal (e.g. \cite{eriksson2021dollar,kirkby2020analysis,knight1992nobody,leggio2003empirical,trainor2005within,williams1993lump}). A comparison is made between the DCA and LS return distributions, and whichever has the ``best'' return distribution is deemed the winner. Here, ``best" is determined via some risk-return assessment. In general, LS offers higher mean returns at the cost of higher standard deviation of returns \cite{williams1993lump}. With respect to Sharpe ratio, Sortino ratio and expected returns, LS is generally superior \cite{eriksson2021dollar,leggio2003empirical}. Using expected utility with constant relative risk aversion, LS is deemed superior in the formulation of \cite{knight1992nobody}, yet inferior in the formulation of \cite{kirkby2020analysis}. Using the first-passage time probability, DCA offers considerably better protection against large losses \cite{trainor2005within}. 

A comparison of DCA and LS where the total amount invested is not equal, but the total length time invested is equal, has been considered in \cite{rozeff1994lump}. In particular, the size of LS investment is adjusted so that it has the same mean terminal wealth as DCA. Then the variance of terminal wealth is compared. Similarly, the size of LS investment is adjusted so that it has the same variance of terminal wealth as DCA. Then the mean terminal wealth is compared. In the former case, LS offers the lower variance, and in the latter case, LS offer the higher mean. So LS is mean-variance dominant over DCA in this sense.

Another, less researched, way to compare DCA with LS involves holding the terminal wealth distribution constant, or at least approximately constant. Then the goal is to determine which investing schedule, DCA or LS, is ``best" in terms of total time invested and size and timing of investments. Here, ``best" is determined based on a particular investor's preference to make several smaller investments over a period of time, versus one lump sum investment over a shorter period of time. The log-normal lower bound presented here facilitates such a comparison.

More generally, it is desirable to compare DCA with LS when the total invested, total time invested, and/or return distribution is not constant. Then the goal is to determine which investing schedule, DCA or LS, is ``best" in terms of total time invested, size and timing of investments, and return distribution. The determination of ``best" can vary between investors, depending on their preferences. However, each investor's determination is likely dependent on the return distribution. Therefore, an expression, or at least an estimation, of the return distribution is needed. The log-normal lower bound presented here offers such an estimation.

Obtaining a simple theoretical expression of the return distribution for an investing schedule where investment amounts and times are predetermined is difficult. Complications arise because placing multiple investments at different times halts application of nice mathematical properties. The distribution has, however, been reasonably described in some special cases. The simplest case is with LS returns, which are often modeled using a log-Normal distribution. A more complicated case is continuous DCA. Using an asset price following geometric Brownian motion w.r.t. time, \cite{milevsky2003continuous} showed that the return for continuous DCA is integrated geometric Brownian motion. Note that the density of integrated geometric Brownian motion has a complicated expression, see \cite{dufresne2001integral,schroder2003integral}. Still, a simple theoretical expression of the return distribution for DCA remains elusive. It is possible to describe the return distribution empirically, like in \cite{kirkby2020analysis}, but a closed form expression is desirable. Considering the difficulty of this problem, an estimate of the return distribution is worth pursuing. In particular, a lower bound is desirable because it will provide a distribution that is no better than the actual return distribution. Then, if investors like the distribution of the lower bound, they will certainly like the actual return distribution.

To make this problem of finding a simple lower bound on returns more manageable, the assumption is made that returns between investment times are log-Normal. Then the return for a given predetermined investing schedule is a sum of dependent log-Normal random variables. From here, the goal is to find a good lower bound for this sum that has a standard distribution, preferably log-Normal. There are several ways to estimate the sum of log-Normal random variables as a log-Normal random variable \cite{beaulieu2004optimal,mehta2007approximating,fenton1960sum,schwartz1982distribution}. However, these estimates are not necessarily lower bounds. A lower bound is presented in \cite{beaulieu1994comparison}, but it lacks a standard distribution and requires independence between elements of the sum. A lower bound that allows dependence between elements of the sum is given by \cite{dhaene2002concept,dhaene2002conceptT}, but it is a lower bound in the convex order sense. The lower bound presented here achieves the goal, as it allows dependence in the summands and has a log-Normal distribution.

\subsection{Main Results}
Here, investment in exactly one asset is considered. Under the assumption that returns between time steps are independent random variables, recursive and closed form expression is given for the positive integer moments of returns (see Theorems \ref{recursion2} and \ref{closedform}). Under the additional assumption that the asset wealth process is geometric Brownian motion, a log-normal lower bound is given on the returns (see Theorems \ref{tG}, \ref{t1} and \ref{tcont}). Theorem \ref{tG} provides recursive expression of the lower bound for the general investing schedule. Theorems \ref{t1} and \ref{tcont} provide closed form expression of the lower bound for DCA and continuous DCA, respectively. Theorem \ref{terror} describes the error and an upper bound on the log-error for the lower bound.

In order to compare DCA with LS when the terminal wealth distribution is approximately constant, the lump sum discount is introduced (see Definition \ref{def:lsd}). The lump sum discount indicates the lump sum investment, in terms of size and time invested, needed to match the distribution of the terminal wealth lower bound for a given investing schedule. Although the given investing schedule need not be DCA, DCA is of particular interest because of its history being compared with LS. Theorem \ref{lslim} describes the limit of the lump sum discount when DCA is the given investing schedule. 

\subsection{Applications}
Results involving the lower bound of returns are applied to data from the S\&P Composite Index, specifically the annual data from 1871 to 2020, including reinvested dividends and adjustment for inflation. This S\&P index tracks a diversified portfolio of large US companies. In particular, it tracks the weighted average of stock prices for the largest US companies, where each company's weight is equal to its market capitalization. Funds tracking this index are very popular among investors. Data is taken from Robert Shiller's online data library, and further details on this data can be found in Section \ref{sec51}. 

First, the S\&P annual price data is fit to a geometric Brownian motion. Then, the main results are applied. Application of the lump sum discount shows a LS investing schedule that is no better than DCA. Outside of the lump sum discount, the goal of applications is not to produce a judgment on whether one investing schedule is better than another. Rather, the purpose is to evaluate how well the lower bound estimates returns and provide investors with informing figures to help decide which investing schedule best suits their preferences. The following two paragraphs discuss specific applications.

Theorem \ref{tG} is used to plot the .025, .5, and $.975$ quantiles for the lower bound of DCA returns and log-returns. The risk of incurring a loss is less than 2.5\% when the total length of investment is over 40 years. The Sharpe ratio for the lower bound of DCA log-returns is also plotted. Theorem \ref{terror} is used to plot the error and log-error for the lower bound of DCA returns. The error and log-error grow with time, and the lower bound is a better estimate of DCA returns when the number of years invested is less than 20. Next, the lump sum discount is illustrated for DCA with varying periods of investment.  When DCA is executed for less than 50 years, a one-time investment that is slightly larger than the total DCA investment, held for less than $\frac{2}{3}$ of the time, produces a terminal wealth distribution that is no better than the DCA. Provided the error is low, meaning the length of investment is less than 20 years, investors can achieve similar terminal wealth to DCA by using LS with a slightly larger total investment and less than $\frac{2}{3}$ the total time invested. 

Since Theorem \ref{tG} is general and need not be confined to DCA, it is applied to a hybrid of LS and DCA. In this hybrid, all investments after the first follow traditional DCA. The initial investment is allowed to vary. The .025, .5 and .975 quantiles are plotted for the lower bound of the hybrid log-returns. In addition, Theorem \ref{terror} is used to plot the error and log-error for the lower bound. Again, the lower bound is a better estimate of the hybrid returns when the number of years invested is less than 20.

\subsection{Organization}
Section \ref{sec2} provides the problem setup and main results. In particular, DCA is formally expressed, assumptions are given and theorems are stated. Section \ref{sec3} describes the data used for applications, validates assumptions needed to apply the theorems given in Section \ref{sec2} and applies those theorems using the data. Section \ref{sec4} provides closing remarks, including a discussion of related future research ideas. Appendix \ref{secA1} provides proofs of the theorems stated in Section \ref{sec2}.

\section{Definitions \& Main Results}\label{sec2}
\subsection{Notation}
Let $X:[0,\infty)\times\Omega\to(0,\infty)$ be a stochastic process defined on the probability space $(\Omega,\mathcal{F},\mathbb{P})$. When using $X(t,\omega)$ as a function of $\omega$ only, write $X(t)$ in place of $X(t,\omega)$. Let $\{t_k\}_{k=0}$ be a sequence in $[0,\infty)$ with $t_0=0$. Suppose that $c_k$ is invested at each time step $t_k$. Let $X_k=\frac{X(t_k)}{X(t_{k-1})}$ for $k=1,2,...$. The returns at time step $t_k$ are given by $R_k=(\sum_{j=0}^{k-1}c_j)^{-1}Y_k$ for $k=1,2,...$, where the $Y_k$ are computed recursively via
\begin{equation}
\begin{split}
Y_1&=c_0X_1,\\
Y_k&=X_k\cdot(Y_{k-1}+c_{k-1}),\quad k=2,3,...
\end{split}
\label{recursion}
\end{equation}
Note that $Y_k$ is the wealth at time $t_k$. The following abbreviations will be used to shorten descriptions: dollar cost averaging (DCA), lump sum (LS), upper bound (UB) and lower bound (LB). In all expressions, $\log$ indicates the natural logarithm.

\subsection{Assumptions}
The results require two assumptions.
\begin{unenumerate}
\item (i)$\quad$The $X_k$ are independent for $k=1,2,...$. 
\item (ii)$\ \ X(t)$ is geometric Brownian motion with $\log X(t)\sim\mathcal{N}(\mu t,\sigma^2t)$, where 
\item $\qquad\mu\in\mathbb{R}$, $\sigma>0$ and $t\geq0$. 
\end{unenumerate}
(i) is used to express the moments of returns. (i) and (ii) are used to obtain a log-normal lower bound on returns. 

\begin{remark}
\normalfont While the investing schedule considered here involves discrete investments, the continuous process of geometric Brownian motion is used because it provides nice structure when investment times are not equidistant. In particular, (ii) makes it easy to describe the distribution of $X_k$:
\begin{equation*}
\log X_k\sim\mathcal{N}(\mu(t_k-t_{k-1}),\ \sigma^2(t_k-t_{k-1})),\quad k=1,2,...
\end{equation*} 
\end{remark}

\subsection{Main Results}
Theorems \ref{recursion2} and \ref{closedform} require assumption (i). Theorem \ref{recursion2} offers recursive expression for the moments of $R_k$. Theorem \ref{closedform} offers closed form expression for the moments of $R_k$. The recursive expression is easy to compute and ideal for applications.

\begin{theorem}
For $n\in\mathbb{N}$, the $n$th moment of returns is given by $\mathbb{E}[R_k^n]=(\sum_{j=0}^{k-1}c_j)^{-n}\mathbb{E}[Y_k^n]$, where the $\mathbb{E}[Y_k^n]$ is found recursively via
\begin{equation*}
\begin{split}
\mathbb{E}[Y_1^n]&=c_0^n\mathbb{E}[X_1^n]\\
\mathbb{E}[Y_k^n]&=\mathbb{E}[X_k^n]\sum_{j=0}^n\binom{n}{j}c_{k-1}^{n-j}\mathbb{E}[Y_{k-1}^j],\quad k=2,3,...
\end{split}
\end{equation*}
\label{recursion2}
\end{theorem}

\begin{theorem}
For $n\in\mathbb{N}$, the $n$th moment of returns is given by 
\begin{equation*}
\mathbb{E}[R_k^n]=\frac{n!}{(\sum_{j=0}^{k-1}c_j)^{n}}\mathbb{E}[X_k^n]\sum_{0\leq j_1\leq j_2\leq...\leq j_k=n}\frac{1}{j_1!}\prod_{l=1}^{k-1}\frac{c_{l}^{j_{l+1}-j_l}\mathbb{E}[X_{l}^{j_l}]}{(j_{l+1}-j_l)!}.
\end{equation*}
\label{closedform}
\end{theorem}

All remaining theorems require assumptions (i) and (ii). Definition \ref{defz} recursively constructs the lower bound of returns. Theorem \ref{tG} describes the lower bound of returns recursively. Under the conditions $c_k=1$ and $t_k=k$ for $k=0,1,...$, Theorem \ref{t1} describes the lower bound of DCA returns in closed form. Note that setting $c_k=1$ and $t_k=k$ for $k=0,1,...$ ensures that 1 unit is invested at non-negative integer times. No other investments are made. Thus, a constant investment amount is made at equidistant time steps. Theorem \ref{tcont} provides an upper bound on the cumulative distribution function of returns when $c_k=\frac{1}{n}$ and $t_k=\frac{k}{n}$ for $k=0,1,...,n-1$, and the limit is taken as $n\to\infty$. This is the continuous version of DCA.

\begin{definition}
Set $Z_1=R_1$ and 
\begin{equation*}
Z_k=\frac{X_k}{\sum_{j=0}^{k-1}c_j}(a_k+c_{k-1})\Bigg(\frac{Z_{k-1}}{a_k}\sum_{j=0}^{k-2}c_j\Bigg)^{\frac{a_k}{a_k+c_{k-1}}},\quad k=2,3,...,
\end{equation*}
where $a_k=\exp\mathbb{E}\Big[\log\big(Z_{k-1}\sum_{j=0}^{k-2}c_j\big)\Big]$. 
\label{defz}
\end{definition}

\begin{remark}
$R_k$ and $Z_k$ are invariant to a rescaling of the investment amounts. In particular, if the substitution $c_j\leftarrow \lambda c_j$ is made for each $j$, then $R_k$ and $Z_k$ do not change.
\end{remark}

\begin{theorem}
For every $k=1,2,...$, $Z_k\leq R_k$ w.p.1 and $\log Z_k\sim\mathcal{N}(m_k,v_k)$, where $m_k$ and $v_k$ are given recursively via 
\begin{equation*}
\begin{split}
m_1&=\mu t_1,\quad v_1=\sigma^2t_1,\\
a_k&=\exp(m_{k-1})\sum_{j=0}^{k-2}c_j,\quad b_k=\frac{a_k}{a_k+c_{k-1}}\\
m_k&=\log\frac{a_k+c_{k-1}}{\sum_{j=0}^{k-1}c_j}+\mu(t_k-t_{k-1}),\\
v_k&=b_k^2v_{k-1}+\sigma^2(t_k-t_{k-1}).
\end{split}
\end{equation*}
\label{tG}
\end{theorem}

\begin{theorem}
Suppose $\mu\neq0$, $c_k=1$ and $t_k=k$ for $k=0,1,...$. Then for every $k=1,2,...$, $Z_k\leq R_k$ w.p.1 and $\log Z_k\sim\mathcal{N}(m_k,v_k)$, where 
\begin{equation*}
\begin{split}
m_k&=\mu+\log\frac{\exp(\mu k)-1}{k(\exp\mu-1)}\\
v_k&=\sigma^2\cdot\frac{-1+\exp(\mu k)[2(1+\exp\mu)+\exp(\mu k)\big(-1-k+(k\exp\mu-2)\exp\mu\big)]}{(\exp(\mu k)-1)^2(\exp(2\mu)-1)}.
\end{split}
\end{equation*}
Furthermore, $v_k$ is increasing with $k$. If $\mu>0$, then $m_k$ is increasing with $k$. If $\mu<0$, then $m_k$ is decreasing with $k$. 
\label{t1}
\end{theorem}

\begin{theorem}
Let $\mathcal{R}_n$ denote the returns at $t=1$ given $c_k=\frac{1}{n}$ and $t_k=\frac{k}{n}$ for $k=0,1,...,n-1$. Then $\mathbb{P}(Z\leq x)\geq\lim_{n\to\infty}\mathbb{P}(\mathcal{R}_n\leq x)$ for all $x\in\mathbb{R}$, where 
\begin{equation*}
\log Z\sim\mathcal{N}\Big(\log\frac{\exp\mu-1}{\mu},\ \sigma^2\cdot\frac{(2\mu-3)\exp(2\mu)+4\exp\mu-1}{2\mu(\exp\mu-1)^2}\Big). 
\end{equation*}
\label{tcont}
\end{theorem}

Observe that $\log Z\sim\mathcal{N}(r_1\mu,\ r_2\sigma^2)$, where
\begin{equation*}
r_1=\frac{1}{\mu}\log\frac{\exp\mu-1}{\mu},\quad r_2=\frac{(2\mu-3)\exp(2\mu)+4\exp\mu-1}{2\mu(\exp\mu-1)^2}.
\end{equation*}
It is not hard to see that $r_1$ is increasing in $\mu$, $\lim_{\mu\to0^+}r_1=\frac{1}{2}$, and $\lim_{\mu\to\infty}r_1=1$. Similarly, $r_2$ is increasing in $\mu$, $\lim_{\mu\to0^+}r_2=\frac{1}{3}$, and $\lim_{\mu\to\infty}r_2=1$. The reason behind studying $r_1$ and $r_2$ is to describe a log-normal return that is no better than the continuous DCA of Theorem \ref{tcont}. When the total investment is equal and $r_2\leq r_1$, a lump sum investment for $r_2$ time units is no better than the continuous DCA of Theorem \ref{tcont}. Figure \ref{fig:r1r2} depicts $r_1$ and $r_2$ over $\mu$.

\begin{figure}[H]
  \includegraphics[width=\linewidth]{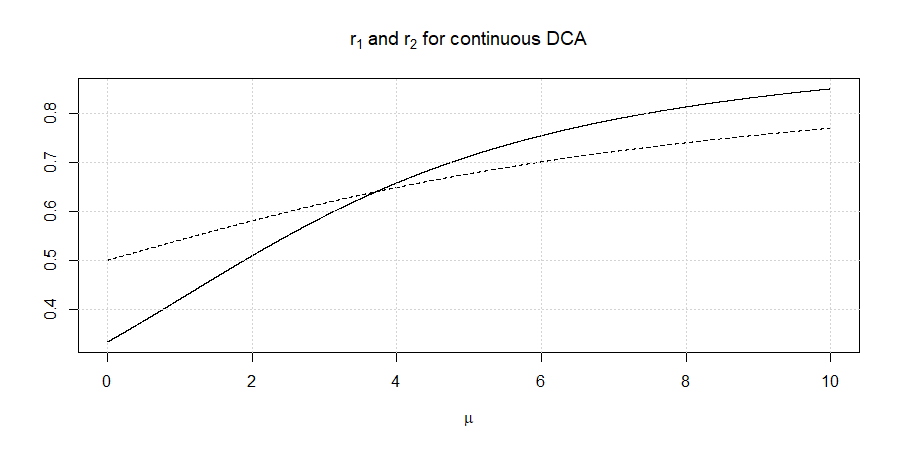}
  \caption{$r_1$ (dashed) and $r_2$ (solid), where $\mu\in(0,10]$.}
  \label{fig:r1r2}
\end{figure}

Theorem \ref{terror} describes the error and an upper bound on the log-error for the lower bound of returns. Furthermore, the upper bound on the log-error is also an upper bound on the relative error. 
\begin{theorem}
Define the error $E_k:=\lvert R_k-Z_k\rvert$ and let $b_k=\frac{\exp(\mu k)-\exp\mu}{\exp(\mu k)-1}$ for $k\in\mathbb{N}$. The expected error is given by 
\begin{equation*}
\mathbb{E}[E_k]=\mathbb{E}[R_k]-\exp(m_k+\frac{v_k}{2}).
\end{equation*}
Define the log-error $E^{\log}_k:=\lvert\log R_k-\log Z_k\rvert$. The expected log-error has upper bound
\begin{equation*}
\mathbb{E}[E^{\log}_k]\leq\log y+\sum_{j=1}^J\frac{\mathbb{E}[(R_k-y)^j]}{(-1)^{j-1}jy^j}-m_k,
\end{equation*}
where $y\in(0,\infty)$ and $J\in\{1,3,5,...\}$. Furthermore, the expected relative error $\mathbb{E}[\frac{E_k}{R_k}]$ has upper bound $\mathbb{E}[E^{\log}_k]$. \\
Expressions for $\mathbb{E}[R_k^j]$, $m_k$ and $v_k$ are given in Theorems \ref{recursion2}, \ref{closedform}, \ref{tG} and \ref{t1}.
\label{terror}
\end{theorem}

Schedules like DCA support smaller investments in periodic increments. This structure fits well with investors who receive monthly or biweekly paychecks. Lump sum investing is ideal, but investors may not have a large amount to invest at once. The lump sum discount addresses the question of how much the strategy of investing smaller amounts over time costs relative to the lump sum strategy. In particular, relatively how much and for relatively how long would a lump sum investment need to be implemented in order to achieve the same wealth distribution as the lower bound for wealth at time $t_k$, which is $Z_k\sum_{j=0}^{k-1}c_j$.

\begin{definition}
The lump sum discount at time step $t_k$ is given by 
\begin{equation*}
\Bigg(\frac{x_k}{\big(\sum_{j=0}^{k-1}c_j\big)^{-1}},\ \frac{s_k}{t_k}\Bigg),
\end{equation*}
where $s_k=\sigma^{-2}v_k$ and $x_k=\big(\sum_{j=0}^{k-1}c_j\big)\exp(m_k-\mu s_k)$. 
\label{def:lsd}
\end{definition}

The lump sum discount is constructed by setting the distribution of $Z_k\sum_{j=0}^{k-1}c_j$ equal to the terminal wealth distribution of a lump sum investment. Both are log-normal under the geometric Brownian motion assumption, so it suffices to set their log-means and log-variances equal to each other. In particular, denote the initial lump sum investment with $x_k$ and the length of investment with $s_k$. Then set $\log(x_k)+\mu s_k=\log\big(\sum_{j=0}^{k-1}c_j\big)+m_k$ and $s_k\sigma^2=v_k$.

To understand the lump sum discount, consider this example for DCA, using the investing schedule given in Theorem \ref{t1}. A lump sum discount of $(\frac{1}{2},\frac{4}{5})$ at time 20 indicates that investing half as much in lump sum for $\frac{4}{5}$ as long achieves the same wealth distribution as the DCA lower bound. When relative error is low, the actual DCA wealth is close to its lower bound. In such situations, the lump sum discount indicates the lump sum investment that will give approximately the same wealth distribution as the DCA. Even when relative error is high, the lump sum discount always indicates a lump sum investment that will perform no better than the DCA. So if an investor is deciding between DCA and the lump sum investment indicated by its lump sum discount, the investor should choose DCA. 

Theorem \ref{lslim} describes the limiting behavior of the lump sum discount for DCA, and Figure \ref{fig:limxk} shows how $\lim_{k\to\infty}x_k$ depends on $\mu$. 

\begin{theorem}
Suppose $\mu>0$, $c_k=1$ and $t_k=k$ for $k=0,1,...$. Then 
\begin{equation*}
\lim_{k\to\infty}x_k=\frac{\exp\Big(\mu\cdot\frac{\exp(2\mu)+2\exp\mu}{\exp(2\mu)-1}\Big)}{\exp\mu-1},
\end{equation*}
and the lump sum discount has limit $\lim_{k\to\infty}(\frac{x_k}{k},\frac{s_k}{k})=(0,1)$. 
\label{lslim}
\end{theorem}

\begin{figure}[H]
  \includegraphics[width=\linewidth]{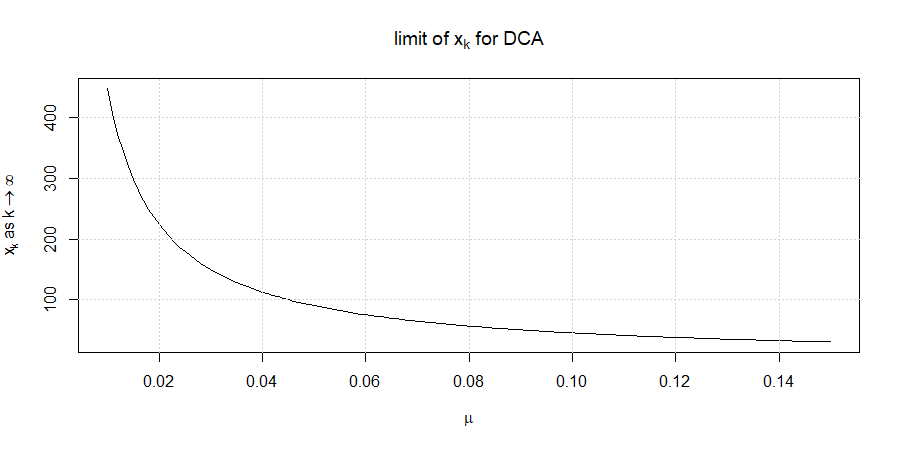}
  \caption{Using Theorem \ref{lslim}, $\lim_{k\to\infty}x_k$ is illustrated for $\mu\in[.01,.15]$.}
  \label{fig:limxk}
\end{figure}

\section{Applications \& Data Analysis}\label{sec3}
\subsection{Data}\label{sec51}
Main results are applied to annual data from the S\&P Composite Index from 1871 to 2020. The data was taken from \url{http://www.econ.yale.edu/~shiller/data.htm} and is collected for easy access at \url{https://github.com/HaydenBrown/Investing}. Here, S\&P Composite Index refers to three indexes: Cowles and Associates from 1871 to 1926, Standard \& Poor 90 from 1926 to 1957 and Standard \& Poor 500 from 1957 to 2020. All were developed to summarize the US stock market in terms of market capitalization \cite{wilson2002}. Companies are given different weights in an index based on their individual market capitalization. The S\&P 90 consists of 90 companies and the S\&P 500 consists of 500 companies. For an overview of the S\&P 500, see \url{https://www.spglobal.com/spdji/en/indices/equity/sp-500/}. The Cowles and Associates index is a backward extension of the S\&P 90 index. Relevant variables from the data are described below. 
\begin{table}[h]
\begin{center}
\caption{Data variable descriptions}
\begin{tabular}{ |c|l| } 
\hline
\textbf{Notation} & \textbf{Description} \\
 \hline
 I & average monthly close of the S\&P composite index \\ 
 \hline
 D & dividend per share of the S\&P composite index \\ 
 \hline
 C & January consumer price index \\ 
 \hline
\end{tabular}
\end{center}
\end{table}
Inflation-adjusted annual returns are computed using the consumer price index, the S\&P Composite Index price and the S\&P Composite Index dividend. Use the subscript $n$ to denote the $n$th year of $C$, $S$ and $D$. Then the inflation-adjusted return for year $n$ is given by $\frac{I_{n+1}+D_n}{I_n}\cdot\frac{C_n}{C_{n+1}}$. From here on, refer to inflation-adjusted returns simply as returns. 

\subsection{Set-up}
The quantile-quantile, autocorrelation function and partial autocorrelation function plots given in Figure \ref{fig:qqacfpacf} show that assumptions (i) and (ii) are acceptable. The QQ plot has some deviations in the tails, but not enough to rule out assumption (ii). In fact, the p-value of the Kolmogorov-Smirnov test is 0.59, so the null hypothesis that log-returns are Normal is not rejected. 
\begin{figure}[h]
  \includegraphics[width=\linewidth]{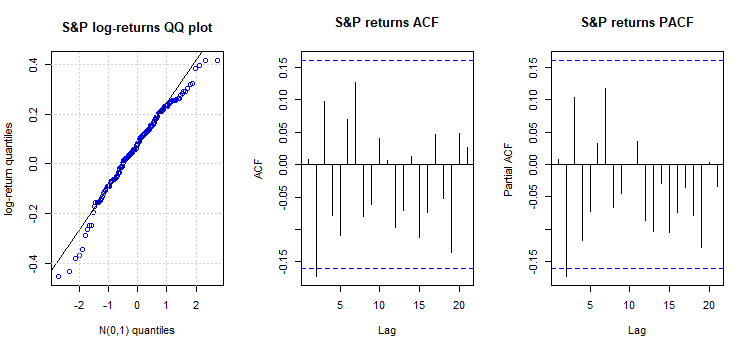}
  \caption{Left: The quantile-quantile plot of annual S\&P log-returns vs the standard Normal distribution. Middle: The autocorrelation function of S\&P returns. Right: The partial autocorrelation function of S\&P returns.}
  \label{fig:qqacfpacf}
\end{figure}
For application of the main results, all that is needed is $\mu$ and $\sigma$. The sample mean of S\&P log-returns is 0.0658, and the sample standard deviation of S\&P log-returns is 0.1690. Therefore $\mu=0.0658$ and $\sigma=0.1690$. Note that the unit on the time domain is years, and applications only use annual investment (i.e. $t_k=k$ for $k=0,1,...$).

\subsection{DCA quantiles, Sharpe ratio and error}
In Figure \ref{fig:bothreturns}, quantiles of the DCA lower bound are produced to give investors an idea of what returns they might experience in terms of the best, worst and median. The scaling on log-returns is easier to read because the 97.5\% quantile for actual returns increases so quickly. Observe the convexity in the lower bound of the 95\% confidence interval for the log-returns. The minimum of the 2.5\% quantile occurs at about 10 years. In contrast, the median and 97.5\% quantiles are strictly increasing. Note that the risk of incurring a loss is less than 2.5\% when the total length of investment is over 40 years. 

\begin{figure}[h] 
  \includegraphics[width=\linewidth]{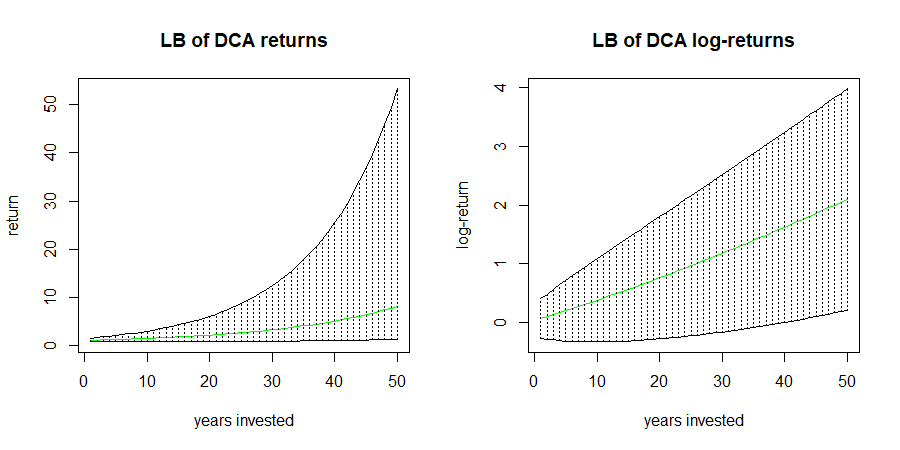}
  \caption{Given $\mu=0.0658$, $\sigma=0.1690$ and $k=1,2,...,50$. The left figure shows $Z_k$ quantiles, and the right figure shows $\log Z_k$ quantiles. The green line indicates the median return. The black lines indicate the upper and lower bound of the 95\% confidence interval.
 }
  \label{fig:bothreturns}
\end{figure}

The Sharpe ratios in Figure \ref{fig:logsharpe} are increasing with time, so investors get a better deal on risk-return as DCA is executed for a longer amount of time. Note the concavity in Sharpe ratios, indicating the risk-return benefit gained by increasing length of DCA investment decreases with time.

\begin{figure}[h]
  \includegraphics[width=\linewidth]{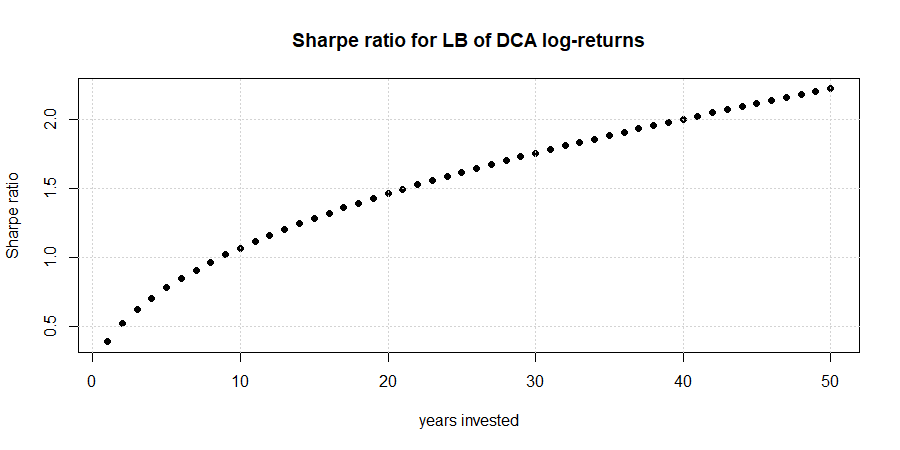}
  \caption{Given $\mu=0.0658$, $\sigma=0.1690$ and $k=1,2,...,50$. The figure shows the Sharpe ratio for the lower bound of log-returns: $\frac{\mathbb{E}[\log Z_k]}{\sqrt{\mathbb{V}[\log Z_k]}}$.}
  \label{fig:logsharpe}
\end{figure}

Figure \ref{fig:botherrors} indicates the expected error of the lower bound increases with time. However, expected returns also increase with time. So it makes sense to consider relative error in order to determine the relative size of the error between returns $R_k$ and their lower bound $Z_k$. The plot of the log-error upper bound, which is also an upper bound for the relative error by Theorem \ref{terror}, shows that the error is expected to be less than $\frac{1}{10}$th of the actual returns when time invested is less than 20 years. This indicates $Z_k$ is a good approximation for $R_k$ for shorter times invested. The plot of log-error gives such a high bound on relative error for longer times invested, like 40 years, that $Z_k$ cannot be claimed a good approximation of $R_k$. It is certainly possible that $Z_k$ is a good approximation for $R_k$ when time invested is 40 years, but the bound in Figure \ref{fig:botherrors} is too high to support such a claim.

\begin{figure}[h]
  \includegraphics[width=\linewidth]{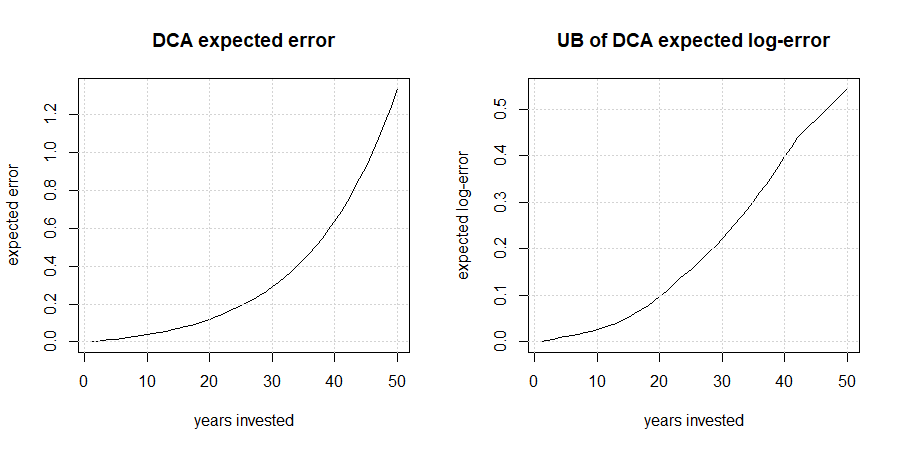}
  \caption{Given $\mu=0.0658$, $\sigma=0.1690$ and $k=1,2,...,50$. The left figure shows $\mathbb{E}[E_k]$, and the right figure shows the upper bound on DCA $\mathbb{E}[E_k^{\log}]$. In computing the log-error, $J$ ranged through the odd numbers from 1 to 21 and for each $k$, $y$ ranged through $\{\exp(-4+.01r)\cdot\mathbb{E}[R_k]:\ r=0,1,2,...,800\}$. See Theorem \ref{terror} for details on $J$ and $y$.}
  \label{fig:botherrors}
\end{figure}

\subsection{Lump sum discount w.r.t. DCA}
Figure \ref{fig:lsd} indicates that a lump sum investment which is slightly larger than the total DCA investment, executed for considerably less time, offers the same wealth distribution as $Z_k\sum_{j=0}^{k-1}c_j$, which is the lower bound of the DCA wealth distribution, $Y_k$. When $Z_k\sum_{j=0}^{k-1}c_j$ is a good approximation to $Y_k$, meaning the total years invested is less than 20 (see Figure \ref{fig:botherrors}), the terminal wealth distribution of such a lump sum investment is close to the actual wealth distribution of the DCA.

\begin{figure}[h]
  \includegraphics[width=\linewidth]{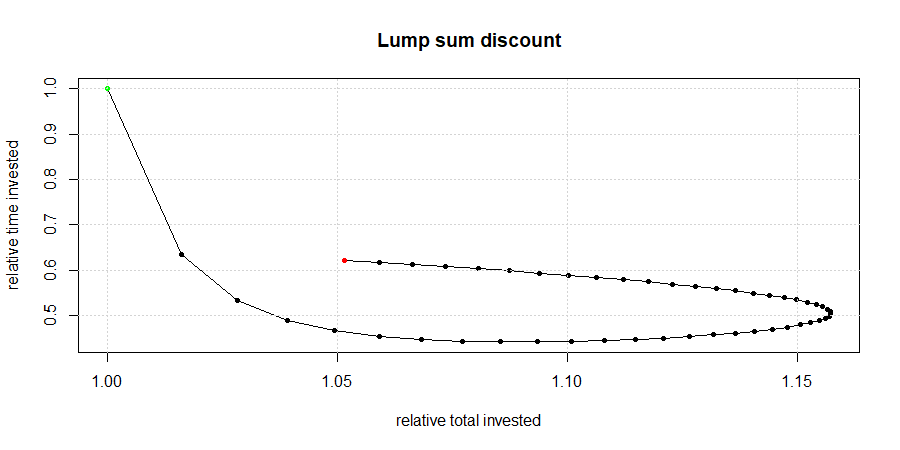}
  \caption{Lump sum discount, given $\mu=0.0658$, $\sigma=0.1690$ and $k=1,2,...,50$. The green point indicates the discount at time 1. Discounts for successive times are found by connecting the points, eventually reaching the discount for time 50, which is indicated by the red point.}
  \label{fig:lsd}
\end{figure}

\subsection{DCA-LS hybrid quantiles and error}
Figures \ref{fig:sametotal}, \ref{fig:Esame}, \ref{fig:diftotal} and \ref{fig:Edif} show quantiles and errors for combinations of LS and DCA. In each example, there is a lump sum investment at year $0$, and then DCA is executed for years 1, 2, 3, etc.. The ratio between the lump sum and DCA investment is varied. 

\begin{figure}[h]
  \includegraphics[width=\linewidth]{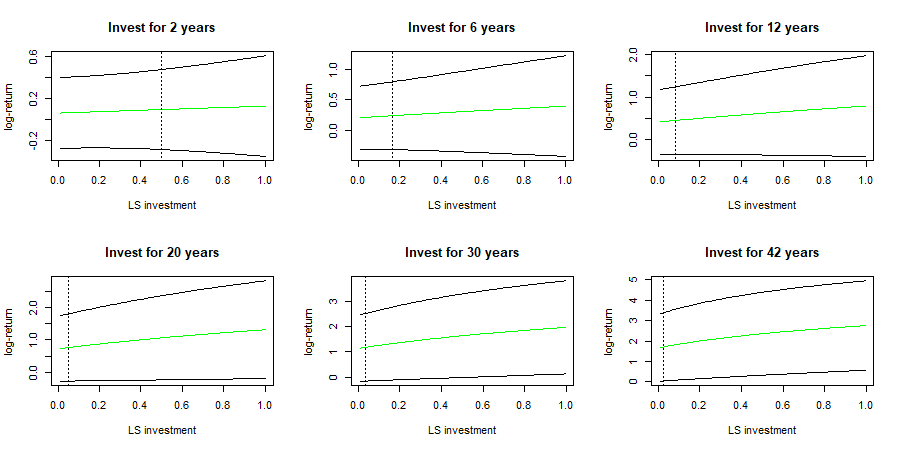}
  \caption{Given $\mu=0.0658$, $\sigma=0.1690$ and $t_{k}=k$ for $k\in\mathbb{N}\cup\{0\}$. Investments are stopped at $T=2,6,12,20,30,42$ years. For each $T$, $c_0$ is given by the horizontal axis, and $c_k=\frac{1-c_0}{T-1}$ for $k=1,2,...,T$. The green line indicates the median return. The black lines indicate the upper and lower bound of the 95\% confidence interval. The vertical dotted line indicates where DCA occurs over the entire sequence of investments, which is when $c_0=c_1$.}
  \label{fig:sametotal}
\end{figure}

\begin{figure}[h]
  \includegraphics[width=\linewidth]{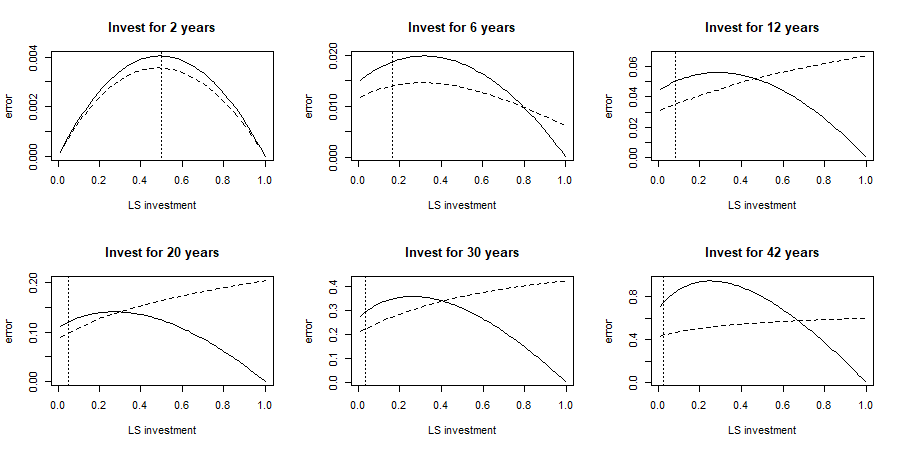}
  \caption{Error and log-error upper bound for Figure \ref{fig:sametotal}. Error is the solid line, log-error upper bound is the dashed line, and the vertical dotted line indicates where DCA occurs over the entire sequence of investments.}
  \label{fig:Esame}
\end{figure}

\begin{figure}[h]
  \includegraphics[width=\linewidth]{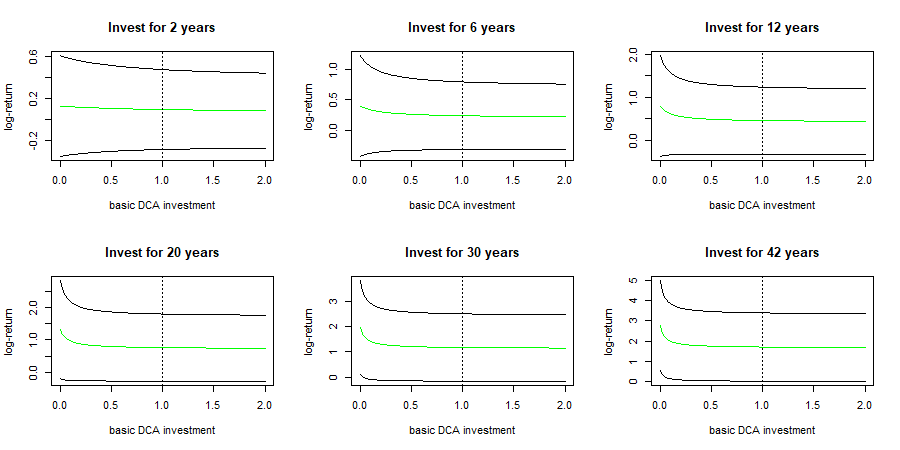}
   \caption{Given $\mu=0.0658$, $\sigma=0.1690$ and $t_{k}=k$ for $k\in\mathbb{N}\cup\{0\}$. Investments are stopped at $T=2,6,12,20,30,42$ years. For each $T$, $c_0=1$; all $c_k$ for $k>0$ are equal and their shared value is given by the horizontal axis. The green line indicates the median return. The black lines indicate the upper and lower bound of the 95\% confidence interval. The vertical dotted line indicates where DCA occurs over the entire sequence of investments, which is when $c_0=c_1$.}
  \label{fig:diftotal}
\end{figure}

\begin{figure}[h]
  \includegraphics[width=\linewidth]{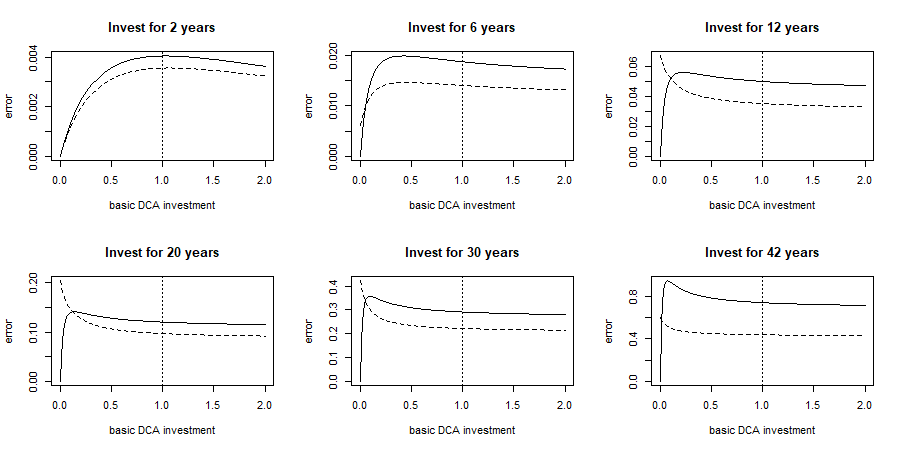}
  \caption{Error and log-error upper bound for Figure \ref{fig:diftotal}. Error is the solid line, log-error upper bound is the dashed line, and the vertical dotted line indicates where DCA occurs over the entire sequence of investments.}
  \label{fig:Edif}
\end{figure}

In Figures \ref{fig:sametotal} and \ref{fig:diftotal}, a lump sum investment is made at time 0, and then DCA is executed for the remaining time. Figure \ref{fig:sametotal} uses an constant total investment of 1 and varies how much of that total investment is made in lump sum at time 0. Figure \ref{fig:diftotal} uses a constant lump sum investment and varies the subsequent DCA investment. Note that the total invested is not constant in Figure \ref{fig:diftotal}. The vertical dotted line in each figure indicates when $c_0=c_1$, meaning DCA is executed over the entire sequence of investments.

In Figure \ref{fig:sametotal}, the .025, .5 and .975 quantiles increase as the investments are executed over a longer period of time. So given a constant lump sum investment, the returns improve as length of investment increases. For investment lengths of at least 20 years, the vertical dotted line, indicating where DCA is executed over the entire sequence of investments, offers nearly the worst return distribution for each investment length. Increasing the lump sum investment increases the .975 and .5 quantiles for all investment lengths. Increasing the lump sum investment increases the .025 quantile for investment lengths of at least 20 years. So investment lengths of at least 20 years have better return distributions when the lump sum investment is increased. Investment lengths less than 20 years offer potential for higher returns when the lump sum investment is increased, but there is the decreasing .025 quantile to factor in. Observe the concavity in quantiles for the 30 and 42 year investments; a small increase in the lump sum investment has a more positive impact on the return distribution when the initial lump sum investment is closer to 0. 

Statements in the previous paragraph assume the lower bound on returns is a good approximation for actual returns. If the error is high, the statements have less support. Figure \ref{fig:Esame} provides the error and log-error upper bound for Figure \ref{fig:sametotal}. The error and log-error are quite low for investment lengths of 2, 6 and 12 years. All lengths of investment have low error when the lump sum investment is close to 1. 

In Figure \ref{fig:diftotal}, the .025, .5 and .975 quantiles increase as the investments are executed over a longer period of time. So given a constant DCA investment, the returns improve as length of investment increases. Observe that the quantiles have asymptotic limits as the DCA investment increases. These limits are simply the DCA lower bound quantiles for the investment length minus 1 year. This is because the lump sum investment of 1 at time 0 becomes very small compared to the DCA investment as the DCA investment is increased arbitrarily. For investments lengths of at least 20 years, the quantiles decrease as the DCA investment increases. So for longer investments, a smaller DCA investment, relative to the lump sum investment, offers the best returns. The situation is similar for investment lengths less than 20 years, but there is the increasing .025 quantile to factor in. Given any investment length, DCA investments greater than 1 produce similar quantiles. Given an investment length of at least 20, DCA investments greater than .25 produce similar quantiles.

Like before, statements in the previous paragraph assume the lower bound on returns is a good approximation for actual returns. If the error is high, the statements have less support. Figure \ref{fig:Edif} provides the error and log-error upper bound for Figure \ref{fig:diftotal}. The error and log-error are quite low for investment lengths of 2, 6 and 12 years. All lengths of investment have low error when the DCA investment is close to 0. Like the quantiles of Figure  \ref{fig:diftotal} the error and log-error have asymptotic limits as the DCA investment increases. These limits are the error and log-error upper bound of the DCA lower bound for the investment length minus 1 year.

\section{Conclusions \& Further Research}\label{sec4}
The lower bound for DCA returns presented here is constructed by taking advantage of the fact that a linear combination of two independent Normal random variables is also Normal. Similar logic can be used to construct a lower bound for DCA returns when the wealth process, $X(t)$, is a L\'evy process, and for each $t$, $X(t)=\exp(\alpha_tY+\beta_t)$, where $\alpha_t,\beta_t\in\mathbb{R}$ and $Y$ is a L\'evy alpha-stable distribution. Then the lower bound for DCA returns will take the form $\exp(\alpha Y+\beta)$, where $\alpha,\beta\in\mathbb{R}$. This more general construction of the lower bound allows for situations where log-returns do not fit a Normal distribution. For example, this construction covers situations where log-returns have heavier tails than the Normal distribution, which is a common occurrence when studying stock price data. 

Here, expression of the moments of returns is only used to approximate the log-error of the lower bound. For the sake of potential future research, note that expression of the moments of returns is especially important in expected utility theory; see \cite{fishburn1970utility}. When the return distribution has compact support, many utility functions can be approximated well with a finite Taylor series. Applying the linearity of expectation to the Taylor series reduces expected utility to a function of the moments of returns. This sort of technique is demonstrated in \cite{conine1981diversification}. Moments of the return distribution can also be used to approximate the return distribution; see \cite{john2007techniques} for a review. 

\begin{appendices}

\section{Proofs}\label{secA1}

\subsubsection*{Theorem \ref{recursion2}}
\begin{proof}
$\mathbb{E}[Y_1^n]=c_0^n\mathbb{E}[X_1^n]$ follows trivially from \eqref{recursion}. Applying independence, the binomial formula and linearity of expectation in successive order,
\begin{equation*}
\begin{split}
\mathbb{E}[Y_k^n]&=\mathbb{E}[X_k^n]\mathbb{E}[(Y_{k-1}+c_{k-1})^n]\\
&=\mathbb{E}[X_k^n]\mathbb{E}[\sum_{j=0}^n\binom{n}{j}c_{k-1}^{n-j}Y_{k-1}^j]\\
&=\mathbb{E}[X_k^n]\sum_{j=0}^n\binom{n}{j}c_{k-1}^{n-j}\mathbb{E}[Y_{k-1}^j].
\end{split}
\end{equation*}
Again by linearity of expectation, $\mathbb{E}[R_k^n]=(\sum_{j=0}^{k-1}c_j)^{-n}\mathbb{E}[Y_k^n]$.
\end{proof}

\subsubsection*{Theorem \ref{closedform}}
\begin{proof}
Here is a proof by induction. The claim is obviously true for $k=1$. Suppose the claim holds for all $k'<k$, where $k\in\mathbb{N}\setminus\{1\}$. Then by Theorem \ref{recursion2} and some algebraic manipulation,
\begin{equation*}
\begin{split}
\mathbb{E}[Y_k^n]&=\mathbb{E}[X_k^n]\sum_{j=0}^n\binom{n}{j}c_{k-1}^{n-j}\mathbb{E}[Y_{k-1}^j]\\
&=\mathbb{E}[X_k^n]\sum_{j=0}^n\binom{n}{j}c_{k-1}^{n-j}(\sum_{i=0}^{k-2}c_i)^{j}\mathbb{E}[R_{k-1}^j]\\
&=\mathbb{E}[X_k^n]\sum_{j=0}^n\binom{n}{j}j!c_{k-1}^{n-j}\mathbb{E}[X_{k-1}^j]\sum_{0\leq j_1\leq j_2\leq...\leq j_{k-1}=j}\prod_{l=1}^{k-2}\frac{c_{l}^{j_{l+1}-j_l}\mathbb{E}[X_{l}^{j_l}]}{j_1!(j_{l+1}-j_l)!}\\
&=n!\mathbb{E}[X_k^n]\sum_{j=0}^n\sum_{0\leq j_1\leq j_2\leq...\leq j_{k-1}=j}\frac{c_{k-1}^{n-j}\mathbb{E}[X_{k-1}^j]}{j_1!(n-j)!}\prod_{l=1}^{k-2}\frac{c_{l}^{j_{l+1}-j_l}\mathbb{E}[X_{l}^{j_l}]}{(j_{l+1}-j_l)!}\\
&=n!\mathbb{E}[X_k^n]\sum_{0\leq j_1\leq j_2\leq...\leq j_{k}=n}\frac{1}{j_1!}\prod_{l=1}^{k-1}\frac{c_{l}^{j_{l+1}-j_l}\mathbb{E}[X_{l}^{j_l}]}{(j_{l+1}-j_l)!}.
\end{split}
\end{equation*}
So the claim holds for $k$ because $\mathbb{E}[R_k^n]=(\sum_{j=0}^{k-1}c_j)^{-n}\mathbb{E}[Y_k^n]$.
\end{proof}

\begin{lemma}
For all $a,c,x\in(0,\infty)$, 
\begin{equation*}
(a+c)\Big(\frac{x}{a}\Big)^{\frac{a}{a+c}}\leq x+c.
\end{equation*}
Moreover, there is equality at $x=a$.
\label{llb}
\end{lemma}
\begin{proof}
The equality at $x=a$ is obvious. Observe that 
\begin{equation*}
\frac{d}{dx}\left[(a+c)\Big(\frac{x}{a}\Big)^{\frac{a}{a+c}}\right]=\Big(\frac{a}{x}\Big)^{\frac{c}{a+c}},\quad \frac{d}{dx}[x+c]=1.
\end{equation*}
So $\frac{d}{dx}\left[(a+c)\Big(\frac{x}{a}\Big)^{\frac{a}{a+c}}\right]<\frac{d}{dx}[x+c]$ when $a<x$ and $\frac{d}{dx}\left[(a+c)\Big(\frac{x}{a}\Big)^{\frac{a}{a+c}}\right]>\frac{d}{dx}[x+c]$ when $a>x$. Thus, $x+c-(a+c)(\frac{x}{a})^{\frac{a}{a+c}}$ is $0$ at $x=a$, decreasing for $0<x<a$, and increasing for $x>a$.  
\end{proof}

\subsubsection*{Theorem \ref{tG}}
\begin{proof}
The result is established via induction. By Definition \ref{defz} $Z_1=R_1=X_1$. Since $\log X_1\sim\mathcal{N}(\mu t_1,\sigma^2t_1)$, it follows that $\log Z_1\sim\mathcal{N}(m_1,v_1)$. 

Now suppose the result holds for all $k'<k$ where $k\in\mathbb{N}\setminus\{1\}$. By \eqref{recursion}, 
\begin{equation}
R_k=\frac{Y_k}{\sum_{j=0}^{k-1}c_j}=\frac{X_k(Y_{k-1}+c_{k-1})}{\sum_{j=0}^{k-1}c_j}=\frac{X_k}{\sum_{j=0}^{k-1}c_j}\Big(R_{k-1}\sum_{j=0}^{k-2}c_j+c_{k-1}\Big).
\label{zpart1}
\end{equation}
By Definition \ref{defz} and log-normality of $Z_{k-1}$, $a_k=\exp\mathbb{E}[\log(Z_{k-1}\sum_{j=0}^{k-2}c_j)]=\exp(m_{k-1})\sum_{j=0}^{k-2}c_j$. By Lemma \ref{llb}, 
\begin{equation}
(a_k+c_{k-1})\Bigg(\frac{R_{k-1}}{a_k}\sum_{j=0}^{k-2}c_j\Bigg)^{\frac{a_k}{a_k+c_{k-1}}}\leq R_{k-1}\sum_{j=0}^{k-2}c_j+c_{k-1},\quad w.p.1.
\label{zpart2}
\end{equation}
Combining \eqref{zpart1}, \eqref{zpart2} and Definition \ref{defz} yields $Z_k\leq R_k$ w.p.1.

Since $X_k$ and $Z_{k-1}$ are independent and log-normal, $Z_k$ must be log-normal. Applying properties of $\log$ to Definition \ref{defz},
\begin{equation}
\log Z_k=\log\frac{a_k+c_{k-1}}{\sum_{j=0}^{k-1}c_j}+\log X_k+b_k\Big(\log Z_{k-1}-\log\frac{a_k}{\sum_{j=0}^{k-2}c_j}\Big).
\label{gzkk}
\end{equation}
From \ref{gzkk}, it is not hard to see that $\log Z_k\sim\mathcal{N}(m_k,v_k)$.
\end{proof}

\subsubsection*{Theorem \ref{t1}}
\begin{proof}
By Theorem \ref{tG}, the result holds when $m_k$ and $v_k$ are as in Theorem \ref{tG}. It remains to be verified is that the expressions of $m_k$ and $v_k$ given in Theorems \ref{tG} and \ref{t1} are equal. This is clearly the case for $k=1$. Suppose the expressions of $m_k$ and $v_k$ given in Theorems \ref{tG} and \ref{t1} are equal for all $k'<k$ where $k\in\mathbb{N}\setminus\{1\}$. Then by Theorem \ref{tG} and the induction assumption,
\begin{equation*}
\begin{split}
m_k&=\log\frac{(k-1)\exp m_{k-1}+1}{k}+\mu\\
&=\log\frac{(k-1)\frac{\exp(\mu(k-1))-1}{(k-1)(\exp\mu-1)}\exp\mu+1}{k}+\mu\\
&=\log\frac{\exp(\mu k)-1}{k(\exp\mu-1)}+\mu.
\end{split}
\end{equation*}
Again by Theorem \ref{tG} and the induction assumption, $v_k=b_k^2v_{k-1}+\sigma^2$, where
\begin{equation*}
b_k=\frac{(k-1)\exp m_{k-1}}{(k-1)\exp m_{k-1}+1}=\frac{\exp(\mu k)-\exp\mu}{\exp(\mu k)-1}.
\end{equation*}
Furthermore, the recursion $v_k=b_k^2v_{k-1}+\sigma^2$ with $v_1=\sigma^2$ implies that $v_k=\sigma^2(1+\sum_{j=0}^{k-2}(\prod_{i=0}^jb_{k-i})^2)$. Observe that for all $i\in\mathbb{N}$,
\begin{equation*}
\frac{\exp(\mu k)-\exp(\mu i)}{\exp(\mu k)-1}\cdot\frac{\exp(\mu (k-i))-\exp\mu}{\exp(\mu (k-i))-1}=\frac{\exp(\mu k)-\exp(\mu(i+1))}{\exp(\mu k)-1}.
\end{equation*}
It follows that 
\begin{equation*}
\prod_{i=0}^jb_{k-i}=\frac{\exp(\mu k)-\exp(\mu (j+1))}{\exp(\mu k)-1}.
\end{equation*}
Thus, for all $k\in\mathbb{N}$,
\begin{equation*}
\begin{split}
\frac{v_k}{\sigma^2}&=1+\sum_{j=1}^{k-1}\Big(\frac{\exp(\mu k)-\exp(\mu j)}{\exp(\mu k)-1}\Big)^2\\
&=1+\frac{\sum_{j=1}^{k-1}\exp(2\mu k)+\exp(2\mu j)-2\exp\{\mu(k+j)\}}{(\exp(\mu k)-1)^2}\\
&=1+\frac{(k-1)\exp(2\mu k)+\frac{\exp(2\mu k)-\exp(2\mu)}{\exp(2\mu)-1}-2\exp(\mu k)\frac{\exp(\mu k)-\exp(\mu)}{\exp\mu-1}}{(\exp(\mu k)-1)^2}\\
&=\frac{-1+\exp(\mu k)[2(1+\exp\mu)+\exp(\mu k)\big(-1-k+(k\exp\mu-2)\exp\mu\big)]}{(\exp(\mu k)-1)^2(\exp(2\mu)-1)}.
\end{split}
\end{equation*}

To see that $m_k$ is increasing/decreasing in $k$ depending on the sign of $\mu$, take the derivative with respect to $k$:
\begin{equation*}
\frac{d}{dk}m_k=\frac{1+\exp(\mu k)(\mu k-1)}{k[\exp(\mu k)-1]}.
\end{equation*}
Since $k[\exp(\mu k)-1]$ has the same sign as $\mu$, it suffices to show that $f(x)=1+\exp x(x-1)$ is positive for all $x\neq0$. Observe that $f(0)=1>0$. So it now suffices to show that $\frac{d}{dx}f(x)=x\exp x$ is negative for $x<0$ and positive for $x>0$. This is obviously the case because $\exp x>0$ whenever $x\neq0$. 

To see that $v_k$ is increasing in $k$, induction is employed. By Theorem \ref{tG}, $v_k=b_k^2v_{k-1}+\sigma^2$, where $b_k=\frac{\exp(\mu k)-\exp\mu}{\exp(\mu k)-1}$. It is not hard to see that $v_2>v_1$. Now suppose $v_k>v_{k-1}$ for some $k\in\mathbb{N}\setminus\{1\}$. Some algebra shows that $b_{k+1}>b_k$. It follows that $b_{k+1}^2v_k>b_k^2v_{k-1}$, and thus $v_{k+1}>v_k$. 
\end{proof}

\subsubsection*{Theorem \ref{tcont}}
\begin{proof}
By Theorem \ref{t1}, for each $n\in\mathbb{N}$, there exists $Z_n$ such that $Z_n\leq \mathcal{R}_n$ w.p.1 and $\log Z_n\sim\mathcal{N}(m_n,v_n)$, using the substitutions $\mu\leftarrow\frac{\mu}{n}$ and $\sigma^2\leftarrow\frac{\sigma^2}{n}$ to evaluate $m_n$ and $v_n$. It follows that $\mathbb{P}(Z_n\leq x)\geq\mathbb{P}(\mathcal{R}_n\leq x)$ for each $x\in\mathbb{R}$ and $n\in\mathbb{N}$. Note that $\lim_{n\to\infty}\mathcal{R}_n$ is integrated geometric Brownian motion w.r.t. time (see \cite{milevsky2003continuous}), so $\lim_{n\to\infty}\mathbb{P}(\mathcal{R}_n\leq x)$ exists. Further, it is not hard to see that $\lim_{n\to\infty}\mathbb{P}(Z_n\leq x)$ exists, as it describes the limit of log-Normal cdfs, whose log-mean and log-variance both converge. In particular, $\lim_{n\to\infty}\mathbb{P}(Z_n\leq x)$ is the cdf of a log-Normal random variable with log-mean $\lim_{n\to\infty}m_n$ and log-variance $\lim_{n\to\infty}v_n$. Note that the limits of $m_n$ and $v_n$ are found using their expressions in Theorem \ref{t1}. More specifically, L'Hopital's rule is used to find $\lim_{n\to\infty}m_n$, and the fact that $\lim_{n\to\infty}n(\exp(\frac{2\mu}{n})-1)=2\mu$ is used to find $\lim_{n\to\infty}v_n$.
\end{proof}

\subsubsection*{Theorem \ref{terror}}
\begin{proof}
By Theorem \ref{tG}, $E_k=R_k-Z_k$ for $k=1,2,...$. Applying linearity of expectation and Theorem \ref{t1},
\begin{equation*}
\mathbb{E}[E_k]=\mathbb{E}[R_k]-\mathbb{E}[Z_k]=\mathbb{E}[R_k]-\exp(m_k+\frac{v_k}{2}).
\end{equation*}
For the log-error,
\begin{equation*}
\mathbb{E}[E_k^{\log}]=\mathbb{E}[\log R_k]-\mathbb{E}[\log Z_k]=\mathbb{E}[\log R_k]-m_k.
\end{equation*}
So it suffices to show that $\mathbb{E}[\log R_k]\leq\log y+\sum_{j=1}^J\frac{\mathbb{E}[(R_k-y)^j]}{(-1)^{j-1}jy^j}$ for $y>0$. First observe that the J-th degree Taylor series expansion of $\log x$ about $y$ is
\begin{equation*}
T_J(x)=\log y+\sum_{j=1}^J\frac{(x-y)^j}{(-1)^{j-1}jy^j}.
\end{equation*}
Fix $J$ odd. Then
\begin{equation*}
\begin{split}
\frac{d}{dx}T_J(x)&=\sum_{j=1}^J\frac{(x-y)^{j-1}}{(-1)^{j-1}y^j}\\
&=\frac{1}{y}\sum_{j=0}^{\frac{J-1}{2}}\Big(\frac{x-y}{y}\Big)^{2j}-\frac{x-y}{y^2}\sum_{j=0}^{\frac{J-3}{2}}\Big(\frac{x-y}{y}\Big)^{2j}\\
&=\frac{1}{y}\cdot\frac{\Big(\frac{x-y}{y}\Big)^{J+1}-1}{\Big(\frac{x-y}{y}\Big)^{2}-1}-\frac{x-y}{y^2}\cdot\frac{\Big(\frac{x-y}{y}\Big)^{J-1}-1}{\Big(\frac{x-y}{y}\Big)^{2}-1}\\
&=\frac{1}{y}\cdot\frac{\Big(\frac{x-y}{y}\Big)^{J+1}-\Big(\frac{x-y}{y}\Big)^{J}+\frac{x-y}{y}-1}{\Big(\frac{x-y}{y}\Big)^{2}-1}\\
&=\frac{1}{y}\cdot\frac{\Big(\frac{x-y}{y}\Big)^{J}+1}{\Big(\frac{x-y}{y}\Big)+1}=\frac{\Big(\frac{x}{y}-1\Big)^{J}+1}{x}.
\end{split}
\end{equation*}
In addition, $\frac{d}{dx}\log x=\frac{1}{x}$. So $\frac{d}{dx}T_J(x)<\frac{d}{dx}\log x$ when $x<y$ and $\frac{d}{dx}T_J(x)>\frac{d}{dx}\log x$ when $x>y$. Since $T_J(y)=\log y $, it follows that $T_J(x)\geq\log x$ for all $x\in(0,\infty)$. Substituting $R_k$ for $x$ gives the result.

For the relative error, observe that $1-x\leq -\log x$ for $x\in(0,\infty)$. Therefore,
\begin{equation*}
\mathbb{E}\Big[\frac{E_k}{R_k}\Big]=\mathbb{E}\Big[1-\frac{Z_k}{R_k}\Big]\leq \mathbb{E}\Big[\log\frac{R_k}{Z_k}\Big]=\mathbb{E}[E_k^{\log}].
\end{equation*}
\end{proof}

\subsubsection*{Theorem \ref{lslim}}
\begin{proof}
Observe that $x_k=k\exp(m_k-\frac{\mu v_k}{\sigma^2})$ and $s_k=\frac{v_k}{\sigma^2}$. From the expression of $v_k$ given in Theorem \ref{t1}, it follows that 
\begin{equation*}
\frac{v_k}{\sigma^2}=k-\frac{2\exp\mu+1}{\exp(2\mu)-1}+\mathcal{O}(\exp(-\mu k)).
\end{equation*} 
Using basic limit properties and the expression of $m_k$ given in Theorem \ref{t1}, it follows that $\lim_{k\to\infty}\frac{v_k}{k\sigma^2}=1$, $\lim_{k\to\infty}\exp(m_k-\frac{\mu v_k}{\sigma^2})=0$ and $\lim_{k\to\infty}x_k$ is as given in the Theorem statement.
\end{proof}

\end{appendices}


\bibliography{sn-bibliography}


\end{document}